\documentclass[10pt]{llncs}
\usepackage{graphicx}
\usepackage[noline,boxruled,linesnumbered,commentsnumbered]{algorithm2e}
\usepackage{amsmath}
\usepackage{amsfonts}
\usepackage{amssymb}
\usepackage{verbatim}

\newtheorem{observation}{Observation}

\begin{document}
\title{Improved Algorithms for the Evacuation Route Planning Problem}
\author{Gopinath Mishra\inst{1} \and Subhra Mazumdar\inst{2} \and Arindam Pal\inst{2}}
\institute{Advanced Computing and Microelectronics Unit\\
Indian Statistical Institute\\
Kolkata, India\\
\email{gopianjan117@gmail.com}
\and Innovation Labs, TCS Research\\
Tata Consultancy Services\\
Kolkata, India\\
\email{\{subhra.mazumdar,arindam.pal1\}@tcs.com}}

\maketitle

\begin{abstract}
Emergency evacuation is the process of movement of people away from the threat or actual occurrence of hazards such as natural disasters, terrorist attacks, fires and bombs. In this paper, we focus on evacuation from a building, but the ideas can be applied to city and region evacuation. We define the problem and show how it can be modeled using graphs. The resulting optimization problem can be formulated as an integer linear program. Though this can be solved exactly, this approach does not scale well for graphs with thousands of nodes and several hundred thousands of edges. This is impractical for large graphs.

We study a special case of this problem, where there is only a single source and a single sink. For this case, we give an improved algorithm \emph{Single Source Single Sink Evacuation Route Planner (SSEP)}, whose evacuation time is always at most that of a famous algorithm \emph{Capacity Constrained Route Planner (CCRP)}, and whose running time is strictly less than that of CCRP. We prove this mathematically and give supporting results by extensive experiments. We also study randomized behavior model of people and give some interesting results.
\end{abstract}

\section{Introduction}
Emergency evacuation is the process of movement of people away from the threat or actual occurrence of hazards such as natural disasters, terrorist attacks, fires and bombs. In this paper, we focus on evacuation from a building, though the ideas can be applied to city and region evacuation. We are motivated by the evacuation drill that regularly happens in our company Tata Consultancy Services. We are developing a system \textsc{SmartEvacTrak} \cite{ahmed2015smartevactrak} for people counting and coarse-level localization for  evacuation of large buildings. Safe evacuation of thousands of employees in a timely manner, so that no one is left behind, is a major challenge for the building administrators. Time is the main parameter in our model. The travel time between different areas of the building is part of the input and the evacuation time is the output. In the following discussion, we use \emph{\{graph, network\}}, \{\emph{node, vertex}\}, \{\emph{edge, arc}\}, and \{\emph{path, route}\} interchangeably.

We have a building along with its floor plan. Employees are present in some portions (rooms) of the building. There are some \emph{exits} on the floor. Every \emph{corridor} has a \emph{capacity}, which is the number of employees that can pass through the corridor per unit time. Every corridor also has a \emph{travel time}, which is the time required to move from the start of the corridor to the end. The goal is to suggest a feasible route for each employee so that he can be guided to an exit. It must be ensured that at any time the number of employees passing through a corridor does not exceed it's capacity.

A complex building does not provide its occupants with all the information required to find the optimal
route. In an emergency, people tend to panic and do not always follow the paths suggested by the algorithm. They are not given enough time to establish a cognitive map of the building. To address this issue, we need to model the behavior of people in emergency situations. We have proposed a simple randomized behavior model and analyzed it. The expected evacuation time comes out to be quite good. None of the previous works considered any behavior model of people.

\section{Related Work}
In this section, we give a summary of different algorithms for the evacuation route planning problem. Skutella \cite{skutella2009introduction} has a good survey on the network flows over time problem. The monograph by Hamacher and Tjandra \cite{hamacher2001mathematical} surveys the state of the art on the mathematical modeling of evacuation problems. Both these papers give a good introduction and comprehensive treatment to this topic.

The LP based polynomial time algorithm for evacuation problem by Hoppe and Tardos \cite{hoppe1994polynomial} uses the ellipsoid method and runs in $O(n^{6}T^{6})$ time, where $n$ is the number of nodes in the graph and $T$ is the evacuation \emph{egress time} for the given network. It uses time-expanded graphs for the network, where there are $T+1$ copies of each node. The expression for time complexity shows that it is not scalable even for mid-sized networks. Another disadvantage is that it requires the evacuation egress time ($T$) \emph{apriori}, which is not easy to estimate. As the time complexity is a function of $T$, it is not a fully polynomial time algorithm.

One of the earliest algorithms by Lu et al. \cite{lu2005capacity} is Capacity Constrained Route Planner (CCRP). CCRP uses Dijkstra's generalized shortest path algorithm to find shortest paths from any source to any  sink, provided that there is enough capacity available on all nodes and edges of the path. An important feature of CCRP is that instead of a single value which does not vary with time, edge capacities and node capacities are modeled as time series (function of time). Here, we need to update edge and node capacities for each time period. The running time of CCRP is $O(p(m + n \log n))$, ($O(pn \log n)$ for sparse graphs, where $m = O(n)$) and space  complexity is $O((m + n)T)$ ($O(nT)$ for sparse graphs). Here $m$ and $n$ denotes the number of edges and the number of vertices of the graph respectively, $p$ denotes the number of evacuees, and $T$ denotes the evacuation egress time. As space complexity is always at most the time complexity, the running time of CCRP is implicitly dependent on $T$. For sparse graphs, $nT \le pn \log n$, \emph{i.e.}, $T \le p \log n$. So, for sparse graphs the evacuation egress time is at most $O(p \log n)$. The space complexity of $O(nT)$ and unnecessary expansion of source nodes in each iteration are two main disadvantages of CCRP.

To overcome the unnecessary expansion in each iteration, Yin et al. \cite{yin2009scalable} introduced the CCRP++ algorithm. The main advantage of CCRP++ is that it runs faster than CCRP. But the quality of solution is not good, because availability along a path may change between the times when paths are reserved and when they are actually used.

\par Min and Neupane \cite{min2011evacuation} introduced the concept of \emph{combined evacuation time} ($CET$) and \emph{quickest paths}, which considers both transit time and capacity on each path and provides a fair balance between them. Let there be $k$ edge-disjoint paths $\{P_{1},P_{2},\ldots,P_{k}\}$  from  source node $s$ to sink node $t$. Then, the combined evacuation time is given by,
	\begin{equation}
	CET(\{P_{1},P_{2},\ldots,  P_{k}\}) = \left\lceil{ \dfrac{p + \sum_{i = 1}^{k}C_{i}T_{i}}{\sum_{i = 1}^{k}C_{i}}}\right\rceil-1 \label{eqn:cet}
	\end{equation}
	where $C_{i}$ and $T_{i}$ denotes the capacity and transit time of path $P_{i}$ respectively, and $p$ denotes the number of evacuees. Time required to evacuate $p$ people via a path $P$ having transit time $T$ and capacity $C$ is $T + \left \lceil \frac{p}{C} \right \rceil - 1$. So, 
$P_{i}$ is said to be the \emph{quickest path} if and only if $T_{i} + \left \lceil \frac{p}{C_{i}} \right \rceil -1 \leq T_{j} + \left \lceil \frac{p}{C_{j}} \right \rceil -1$, for all $j \in \{1, \ldots, k\} \setminus \{i\}$. 
	
The formula for combined evacuation time not only gives an exact expression for the evacuation time, but it also gives the number of people that will be evacuated on each path. The intuition behind the concept of $CET$ is that paths having lesser arrival time will evacuate more groups. This algorithm is known as QPER (Quickest Path Evacuation Routing). The algorithm finds all edge-disjoint paths between a single source and a single sink and orders them according to the quickest evacuation time (calculated using $CET$) and adds them one by one. The algorithm is fairly simple. It does not use time-expanded graphs and there is no need to store availability information at each time stamp, as only edge-disjoint paths are considered. But their algorithm is limited to single source and single sink evacuation problems. Besides these, the addition of paths is not consistent, i.e., a path added at some point of time may be removed by the algorithm at a latter point of time, in case removal makes the solution better.

\par The solutions produced by CCRP++ and QPER do not follow semantics of CCRP, i.e., the solution quality is not better than that of CCRP. Recently Gupta and Sarda \cite{gupta2014efficient} have given an algorithm called CCRP*, where evacuation plan is same as that of CCRP and it runs faster in practice. Instead of running Dijkstra's algorithm from scratch in each iteration, they resume it from the previous iteration.

Kim et al. \cite{kim2008contraflow} studied the contraflow network configuration problem to minimize the evacuation time. In the \emph{contraflow} problem, the goal is to find a reconfigured network
identifying the ideal direction for each edge to minimize the evacuation time, by reallocating the available capacity. They proved that this problem is NP-complete. They designed a greedy heuristic to produce high-quality solutions with significant performance. They also developed a bottleneck relief heuristic to deal with large numbers of evacuees. They evaluated the proposed approaches both analytically and experimentally using real-world data sets. Min and Lee \cite{min2013maximum} build on this idea to design a maximum throughput flow-based contraflow evacuation routing algorithm.

Min \cite{min2012synchronized} proposed the idea of \emph{synchronized flow} based evacuation route planning. Synchronized flows replace the use of time-expanded graphs and provides higher scalability in terms of the evacuation time or the number of people evacuated. The computation time only depends on the number of source nodes and the size of the graph.

Dressler et al. \cite{dressler2010use} uses a network flow based approach to solve this problem. They use two algorithms: one is based on \emph{minimum cost transshipment} and the other is based on \emph{earliest arrival transshipment}. They evaluate these two approaches using a cellular automaton model to simulate the behavior of the evacuees. The minimum cost approach does not consider the distances between evacuees and exits. It may fail if there are exits very far away. Problems also arise if a lot of exits share the same bottleneck edges. The earliest arrival approach uses an optimal flow over time and thus does not suffer from these problems. But the exit assignment computed by the earliest arrival approach may not be optimal.

There are some previous works which considered the behavior of people in an emergency. L{\o}vs \cite{lovs1998models} proposed different models of finding escape routes in an emergency. Song et al. \cite{AAAI159418} collect big and heterogeneous data to capture and analyze human emergency mobility following different disasters in Japan. They develop a general model of human emergency mobility using a Hidden Markov Model (HMM) for generating or simulating large amount of human emergency movements following disasters.

\begin{figure}
\begin{center}
\includegraphics[scale=0.55]{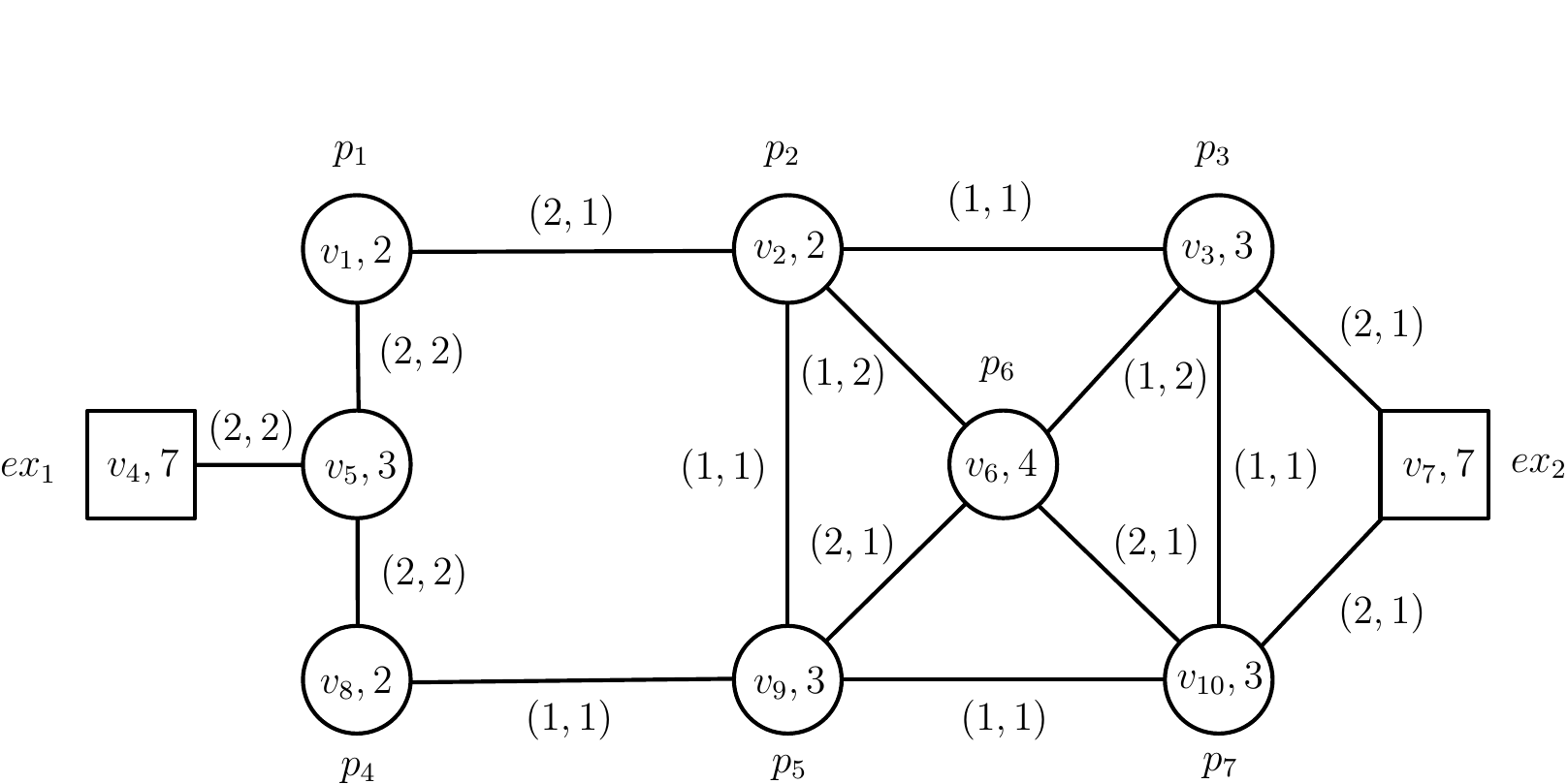}
\caption{A building graph, where vertices represented as squares denote exits. 
The vertex name and capacity are written inside a vertex. 
The edge capacity and travel time are written beside an edge. 
Persons residing on a vertex are specified beside that vertex.}
\label{example-graph1}
\end{center}
\end{figure}

\section{Problem Definition and Model}
The building floor plan can be represented as a graph $G=(V,E)$, where $V$ and $E$ are the set of vertices and edges respectively. The number of vertices and edges are $n$ and $m$ respectively. Nodes represent rooms, lobbies and intersection points and arcs represent corridors, hallways and staircases. Some nodes in the building having significant number of people are modeled as \emph{source} nodes. The exits of a building are represented as \emph{sink} nodes. Each node has a \emph{capacity}, which is the maximum number of people that can stay at that location at any given time and an \emph{occupancy}, which is the number of people currently occupying the location. Here, $p$ is the total number of people who needs to be evacuated.

Each edge has a \emph{capacity}, which is the maximum number of people that can traverse the edge per unit time and a \emph{travel time}, which is the time needed to travel from one node to another along that edge.

Figure \ref{example-graph1} shows a building graph that consists of $10$ vertices and $15$ edges. For each vertex $v$, it's name and the capacity are specified by a pair of the form $(v,c(v))$. A vertex representing an exit is drawn as a square, while the others are drawn as circles. For each edge $e$, the capacity and the travel time are specified on the edge by the pair $(c(e),d(e))$. The goal is to find the exit and the path (route) for each employee, subject to the constraint that the number of source-sink paths passing through an edge does not exceed the capacity of the edge at any unit time interval. The objective function we want to minimize is the total time of evacuation, that is the time at which the last employee is evacuated. Let's define this as the \emph{evacuation time}. In the \emph{quickest flow problem}, we are given a flow value $f$. We want to \emph{minimize} the time $T$ in which a feasible flow of value at least $f$ can be sent from sources to sinks.

\section{The Single Source Single Sink Problem}
In this section, we focus on the single source single sink evacuation (SSEP) problem. In real life, single source single sink evacuation problem has many applications. For example, if all the people are in an auditorium, and there is only one exit in the building, we want to evacuate people as soon as possible, when there is an emergency. Throughout the rest of this paper, $s$ denotes the source and $t$ denotes the sink. Before proceeding further let's have some definitions.

\begin{definition}
\emph{Transit time of a path} is the sum of the transit times of all the edges in $P$ from $s$ to $t$, and is denoted as $T(P)$.
\end{definition}

\begin{definition}
\emph{Destination arrival time of a path} is the time required by a person to move from $s$ to $t$ using path $P$ subject to prior reservations, and is denoted as $DA(P)$. In other words, we can say that $DA(P)$ is the sum of $T(P)$ and any intermediate delay. Note that $DA(P) \geq T(P)$. 
\end{definition}

\begin{definition}
\emph{Capacity of a path} is the minimum of the capacities of all nodes and edges present in the path $P$, and is denoted by $C(P)$.
\end{definition}

\begin{definition}
A node (edge) on a path $P$ is called \emph{saturated} if the capacity of the node (edge) equals the capacity of $P$.
\end{definition}

\begin{definition}
Two paths $P_{1}$ and $P_{2}$ are said to be \emph{distinct} if $V_{1} \neq V_{2}$ or $E_{1} \neq E_{2}$, where $V_{1},V_{2}$ are the set of vertices and $E_{1},E_{2}$ are the set of edges on the paths $P_{1}$ and $P_{2}$ respectively.
\end{definition}
	
\subsection{Limitation of QPER Algorithm for SSEP}
Using the concept of combined evacuation time, Min et al. \cite{min2011evacuation} gave an algorithm $QPER$ for the single source single sink evacuation problem. Their algorithm works well when we have already discovered $k$ edge-disjoint paths. In $QPER$, paths from $s$ to $t$ are added one by one in  ascending order of quickest paths, and new $CET$ is calculated after each path addition. But after addition of a path, the new $CET$ may be less than the transit time of a previously added path. In that case, we have to delete those paths which have higher transit time than the current $CET$. This in turn increases the running time, since the addition of paths is not consistent. 
\par We overcome the above limitations of the algorithm by adding paths in increasing order of transit time in each iteration till the transit time of the currently discovered path exceeds the $CET$ of the  previously added set of paths. Note that, we need not discover all possible paths from source to sink, since unlike $QPER$, if a path is added in any iteration, it will remain till the end. The $CET$ after each iteration will be monotonically non-increasing.

\subsection{Modified algorithm for SSEP when we are given $k$ edge-disjoint paths}
\label{sec:modified}		
Let  $P_{1}, P_{2},\ldots, P_{k}$ be $k$ edge-disjoint paths from $s$ to $t$ in ascending order of their transit time, \emph{i.e.}, $T_{1} \leq T_{2} \leq \ldots \leq T_{k}$. We define, $S_i = \{P_1, \ldots, P_i\}$. We add paths to our set of routes ($\cal R$) in the following fashion.
\begin{enumerate}
	\item ${\cal R} = \{P_{1}\}$.
 	\item $CET = CET(S_{1})$.
 	\item Start with $i = 1$ Execute step $4$ and $5$ till $i \leq k$ and $T_{i + 1} \leq CET$.
 	\item Add path $P_{i + 1}$ to {\cal R}.
 	\item $CET = CET(S_{i + 1})$ and $i\leftarrow i + 1$.
 	\item Return ${\cal R}$.
\end{enumerate}

\begin{lemma}\label{lemma:edge_disjoint_1}
 \textit{If $S_{j} = \{P_{1}, P_{2},\ldots, P_{j}\}$, $j \leq k$ is returned as ${\cal R}$ by the above algorithm then \newline
 1. $T_{l + 1} \leq CET(S_{l}),  1 \leq l < j$\newline
 2. $CET(S_{1}) \geq CET(S_{2}) \geq \ldots \geq CET(S_{j})$\newline 
 3. $CET(S_{j}) \leq CET(S_{l}), j < l \leq k$}.
 \end{lemma}

\begin{proof}
Directly follows from the algorithm.
\end{proof}
\begin{lemma}\label{lemma:edge_disjoint_2}
 \textit{ If $S_{j} = \{P_{1},  P_{2},\ldots, P_{j}\}$, $j \leq k$ is returned as ${\cal R}$ by above algorithm then $T_{1} \leq T_{2} \leq \ldots \leq T_{j} \leq CET(S_{j}) \leq CET(S_{j - 1}) \leq \ldots \leq CET(S_{1}) $}
\end{lemma}
\begin{proof}
 Here $T_{1} \leq T_{2} \leq \ldots \leq T_{j}$ and by Lemma \ref{lemma:edge_disjoint_1} $CET(S_{j}) \leq CET(S_{j - 1}) \leq \ldots \leq CET(S_{1})$. So, the only thing remains to prove is $T_{j} \leq CET(S_{j})$. Let by contrary assume that $T_{j} > CET(S_{j})$. By putting formula for $CET(S_{j-1})$ from equation (\ref{eqn:cet})  and then solving we get $T_{j} > CET(S_{j - 1})$. By Lemma \ref{lemma:edge_disjoint_1}, $T_{j} \leq CET(S_{j - 1})$. This is a contradiction.
\end{proof}

\begin{lemma}
 \textit{If $S_{j} = \{P_{1}, P_{2},\ldots, P_{j}\}$, $j \leq k$ is returned as ${\cal R}$ by  above algorithm then $CET(S_{j}) \leq CET(S_{j} \setminus \{P_i\}),$ $ 2 \leq i \leq j$.}
\end{lemma}

\begin{proof}
  We will prove this statement by contradiction. Let $CET(S_{j}) > CET(S_{j} \setminus \{P_i\})$,  which implies  $T_{i}  > CET(S_{j})$ by putting formula for $CET$ from equation-$1$. It is not possible by Lemma \ref{lemma:edge_disjoint_2}. Hence the claim holds.
\end{proof}

\begin{remark}
The addition of paths by the above algorithm is consistent, i.e. if a path is added then it will remain till the end of the algorithm execution.
\end{remark}

\subsection{An Important Observation}
\begin{figure}
\begin{center}
\includegraphics[scale = 0.55]{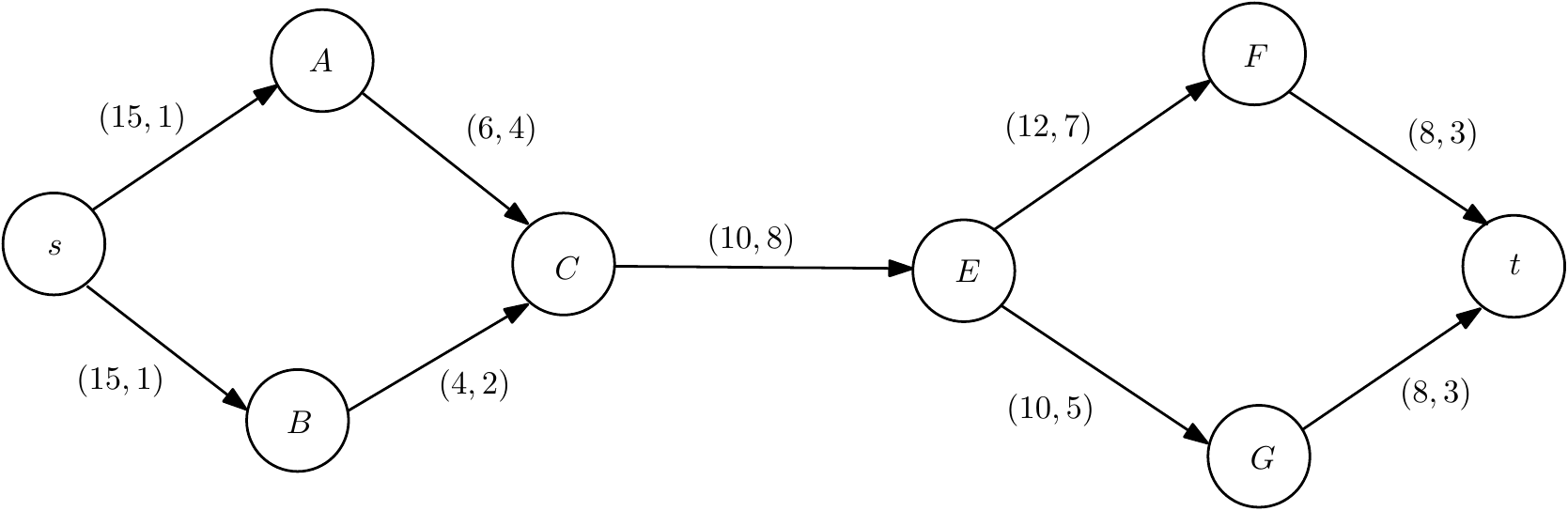}
\caption{An example to show that parallel flows can be sent on non edge-disjoint paths.}
\label{example-graph2}
\end{center}
\end{figure}
In Figure \ref{example-graph2}, ordered pair $(C, T)$ denotes capacity and transit time of an edge.
There are two paths $P_{1}$ and $P_{2}$ between $s$ and $t$.\newline
$P_{1} : s-B-C-E-G-t$, $C(P_{1}) = 4$, $T(P_{1}) = 19$.\newline
$P_{2} : s-A-C-E-F-t$, $C(P_{2}) = 6$, $T(P_{2}) = 23$.\newline
$P_{1}$ and $P_{2}$ are not edge-disjoint, but  common edge $CE$  has capacity of $10$ i.e. $C(P_{1}) + C(P_{2}) = C(CE)$. So, flow can be sent through $P_{1}$ and $P_{2}$ in parallel and we may think like we have two copies of edge $CE$ one having capacity $4$, dedicated for $P_{1}$ and other one having capacity $6$, dedicated for $P_{2}$. We name such set of paths as "virtually edge disjoint". Now it is easy to observe that to apply the formula of combined evacuation time on a set of paths, defined in equation (\ref{eqn:cet}),  the necessary condition is they should be virtually edge disjoint rather than edge disjoint. 

\subsection{Our Algorithm for SSEP}
\par The main idea of the algorithm is to find set of virtually edge disjoint paths one by one and calculate CET  as in section \ref{sec:modified} after each path addition till it satisfies a required condition. 
\par 	 We discover paths one by one in the order of their transit time as follows. We find path $P_{1}$ along with its capacity $C_{1}$ having minimum transit time and decrease capacities of each node and path of $P_{1}$ by its capacity $C_{1}$ permanently and delete saturated nodes and edges. Let's say we have already added paths $\{P_{1},~P_{2},\ldots,~P_{i}\}$, $i \geq 1$, and updated the capacities of nodes and edges along with deletion of required saturated nodes and edges. Note that $P_{1},~P_{2},\ldots,~P_{i}$ are virtually edge disjoint. Hence formula of CET can be applied. In next iteration we discover a path $P_{i+1}$ in residual graph iff $t$ is reachable from $s$ and $i < p$(see line number-4 in algorithm \ref{main_algo}). We add the discovered path $P_{i+1}$ iff $T_{i+1} \leq CET(S_{i})$(see line number-6 in algorithm \ref{main_algo}). As we delete saturated nodes and edges in each iteration when a path is added  we discover paths in maximum of $m+n$ iterations i.e. at max $m+n$ paths and we are not going to discover more than p paths as each path can evacuate atleast one people. So, our algorithm restricts  finding exponential number of possible paths from $s$ to $t$ . More clearly we discover at most $min(m+n, p)$ paths.
\par Here one may think of we are adding paths only based on transit time without considering capacity. Note that  selection of a path for addition is  based on transit time, addition of selected path is done if its transit time less than or equal to  previously calculated CET, which is function of both capacities and transit times of previously added paths. So, our addition of paths to the solution is based on both transit time and capacities of paths implicitly.

\begin{algorithm}
\scriptsize
\caption{Single Source Single Sink Evacuation Route Planner (SSEP)}
\label{main_algo}
\KwIn{A graph $G(V, E)$ representing the network with designated source $s \in V$ and sink $t \in V$. Every node $v \in V$ has an occupancy and maximum capacity. Every edge $e \in E$ has a maximum capacity and transit time. Initially, all persons are in $s$.}
\KwOut{Evacuation route plan for each person.}
\LinesNumbered
\Begin
	{
	Initialize ${\cal R} = \emptyset$ and $CET = \infty$.\\
	Initialize $i\leftarrow 0$.\\
	\While{($t$ is reachable from $s$) and number of discovered paths $\le p-1$ }
	{
		Find the shortest path $P_{i + 1}$ from $s$ to $t$ in $G(V, E)$and let $T_{i+1},C_{i+1}$ be its transit time and capacity respectively.\\
		\If{$T_{i+1} \leq CET$ }
		{
			${\cal R} ={\cal R} \cup \{P_{i+1}\}$.\\
			$CET = CET(S_{i+1})$.\\
			Reduce capacity of each node and each edge of $P_{i + 1}$ by $C_{i+1}$.\\
			$V = V \setminus \{v:v$ is a saturated node of $P_{i + 1}\}$.\\
			$E = E \setminus \{e:e$ is a saturated edge of $P_{i + 1}\}$.\\
		}
		\Else
		{
			break.
		}
	}
    Let ${\cal R} = \{P_{1},P_{2},\ldots,P_{k}\}$. \\
		Send $x_{i}$ persons via $P_{i}, 1 \leq i \leq k$,  where $T_{i} + \lceil \frac{x_{i}}{C_{i}} \rceil - 1 = CET$.
}
\end{algorithm}

\subsection{Running Time Analysis of SSEP}
From the above discussion it is clear that at most $min(m+n, p)$ paths will be discovered and equivalently our algorithm runs for at most $min(m+n, p)$ iterations. As each path discovery can be done in $O((m + n \log n)$ time, using well known Dijkstra algorithm for shortest path, our entire algorithm requires $O(\min(m + n, p)(m + n \log n)$ time. Assuming $m = O(n)$, this becomes $O(\min(n, p) \cdot n \log n)$, which is always at most $O(p n \log n)$. Recall that the time-complexity of CCRP is $O(p n \log n)$. Hence, SSEP always performs faster than CCRP. In real life, the number of evacuees is much larger than the number of vertices, so SSEP runs much faster than CCRP.

\subsection{CCRP Algorithm for SSEP and Some Observations}\label{CCRP}
CCRP \cite{lu2005capacity} is an industry standard algorithm. Many studies have shown that the quality of solution produced by CCRP is better than most heuristic algorithms. We present the CCRP algorithm in simplified form, when there is a single source and a single sink.
\begin{enumerate}
	\item $s$ is added to the priority queue. The nodes in priority queue are ordered based on the distance calculated from $s$ during algorithm execution.
	\item While there are evacuees in $s$, find the path $P$ having minimum destination arrival time from $s$ to $t$ taking the capacity of the various nodes and edges into consideration.
	\item Find  capacity of $P$ and reserve capacity along  the path for a group of size equal to the minimum capacity.
	\item If there are evacuees left at $s$, go to step $2$.
\end{enumerate}
\begin{definition}[\bf{Group Size of a path}]
In each iteration of CCRP one path (say $P_{i}$) from $s$ to $t$ is discovered along with maximum number of people that can be evacuated through that path. This is defined as the \emph{group size} of $P_{i} $ for this iteration.
\end{definition}
For the below sections we denote $T_{i}$, $C_{i}$ as transit time and group size of path $P_{i}$ respectively.

\begin{observation}\label{main_obsv}
\textit{Let's consider execution of single source($s$) single sink($t$) evacuation network by CCRP algorithm. Let $P_{1}, P_{2}, \ldots, P_{k}$ be distinct paths(not necessarily edge-disjoint) from $s$ to $t$  discovered by CCRP such that  $T_{1} \leq T_{2} \leq \ldots  \leq T_{k}$. Here $A_{i}(T)$ is any permutation of $P_{1}(T), P_{2}(T), \ldots, P_{i}(T)$ and $P_{j}(T)$ is the path $P_{j}$ with destination arrival time $T$.\\
Phase $1$: $A_{1}(T_{1}), A_{1}(T_{1} + 1),\ldots, A_{1}(T_{2} - 1)$\\
$\ldots$\\
Phase $i$: $A_{i}(T_{i}), A_{i}(T_{i} + 1),\ldots ,A_{i}(T_{i + 1} - 1), i < k$\\
$\ldots$\\
Phase $k$: $A_{k}(T_{k}), A_{k}(T_{k} + 1),\ldots,A_{k}(T_{k} + \epsilon - 2),A_{k}(T_{k} + \epsilon - 1) $.\\
 Here $\epsilon $ is the maximum number of times any path is discovered in phase $k$. Note that $\epsilon \geq 1$ as $P_{k}$ is discovered at least once. \\
Number of times any path discovered in phase-$k$ is either $\epsilon$ or $\epsilon -1$. It is because of the following argument. By definition of $\epsilon$ there exists a path (say $P_{m}$)  discovered $\epsilon$ number of times. Let $P_{l}$ is a path discovered less than $\epsilon - 1$ number of times. In this case CCRP algorithm would have returned $P_{l}$ instead of $P_{m}$, because using path $P_{l}$ some people can reach destination before or at time $T_{k} + \epsilon - 2$ and  $P_{m}$ has earliest destination arrival time of $T_{k}+\epsilon -1$.\\
Consider the point when all $k$ paths have been returned $\epsilon - 1$ times in phase $k$. Now we may not have enough evacuees such that CCRP will return each path once. We can add some virtual evacuees such that we will use all the paths exactly $\epsilon$ times in phase-$k$ and for simplicity we can say $\epsilon$ is the number of times path $P_{k}$ is returned by CCRP.\\
 Here it is easy to note that evacuation egress time $T_{Evac}^{CCRP} = T_{k} + \epsilon -1$ and it is independent of permutation of paths in any $A_{i}(T)$. So, fix a permutation i.e. \textbf{$A_{i}(T) = P_{1}(T), P_{2}(T), \ldots, P_{i}(T)$}. Fixing up this permutation doesn't affect the solution, but it will make the analysis easier.}
\end{observation}

\begin{observation}\label{discov_shortestpath_obsv}
\textit{Let $P_{1}, P_{2}, \ldots, P_{k}$ be distinct paths(not necessarily edge-disjoint) from $s$ to $t$  discovered by CCRP such that  $T_{1} \leq T_{2} \leq \ldots  \leq T_{k}$. Here $P_{i}$ is the shortest path discovered after deletion of saturated nodes/edges of $P_{1}, P_{2},\ldots,P_{i - 1}$.}
\end{observation}

\begin{remark}
Algorithm \ref{main_algo} finds a path even after we have deleted saturated nodes and edges of all previously discovered path, if it satisfies the conditions given on line numbers 4 and 6.
\end{remark}

\begin{observation}\label{group_obsv}
\textit{Let's consider the sequence of paths as in Observation \ref{main_obsv} with the fixed permutation of each $A_{i}(T)$ as explained. A path $P_{i}$ may be returned in many iterations of CCRP. Group size returned in all iterations are equal possibly except last time when $P_{i}$ is discovered(in phase $k$) in case we don't have enough evacuees left at $s$. This type of situation might happen only once as we are dealing with single source single destination network and it can happen in phase $k$ after or while discovery of $P_{k}$ for the first time. In such cases we can add some virtual evacuees to $s$ so that group size of a path remains same in all iterations. It will not affect evacuation egress time but it will make the analysis easier.}
\end{observation}

\begin{remark}
We can represent each path discovered by CCRP as an ordered pair of path and its group size. Algorithm \ref{main_algo} returns a path with maximum number of people who can travel by that path at any time. As each path is discovered only once, we can represent each path along with the capacity as an ordered pair.
\end{remark}

\subsection{Analysis of Algorithm \ref{main_algo}}
\begin{lemma}\label{analysis_lemma}
Let $(P_{1}, C_{1}), (P_{2},C_{2}), \ldots, (P_{k}, C_{k})$ be distinct paths (not necessarily edge-disjoint) from $s$ to $t$ in order of their transit time discovered by CCRP.
\begin{enumerate}
\item Number of iterations that will return path $P_{i} $ is $T_{k} - T_{i} + \epsilon $, $1 \leq i \leq k $, where $\epsilon$ denotes number of iterations that returns path $P_{k}.$
\item Number of iterations that will return path $P_{i}$ before phase $j$ is  $T_{j} - T_{i}$, where $i \leq j \leq k$.
\item The same paths will be returned by Algorithm \ref{main_algo}, and $T_{1} \leq T_{2} \leq \ldots \leq T_{k}$.
\end{enumerate}
\end{lemma}

\begin{proof}
Parts (1) and (2) directly follows from Observation \ref{main_obsv}. For part (3), by induction we can prove that algorithm \ref{main_algo} finds each path $P_{j}, 1 \leq j \leq k$ with available capacity $C_{j}$.
 
\textbf{Base case:} $j = 1$ i.e. $(P_{1},C_{1})$ is added by Algorithm \ref{main_algo}. This is obvious.

\textbf{Inductive Step:} Suppose paths $(P_{1},C_{1}),\ldots,(P_{j},C_{j}), 1 \leq j < k$ have been added by Algorithm \ref{main_algo}. We have to prove that Algorithm \ref{main_algo} will also add $(P_{j+1},C_{j+1})$.
 
 \textbf{Part 1:} From Observation \ref{discov_shortestpath_obsv}, $P_{j +1 }$ is the shortest path from $s$ to $t$ in residual graph i.e. if we delete saturated node(s) and/or edge(s) of the paths $P_{1},P_{2},\ldots,P_{j}$. Algorithm \ref{main_algo} also adds paths one by one after deleting saturated node(s) and/or edges(s) of previously discovered paths. So,  structure of the graph remains same after addition of these $j$ paths both in CCRP and Algorithm \ref{main_algo}. So, $P_{j + 1}$ is also the best path w.r.t. transit time in residual graph according to Algorithm \ref{main_algo}. As $P_{j + 1}$ is the best path in residual network either no paths will be added or $P_{j + 1}$ will be added to set of routes in Algorithm \ref{main_algo}.
 
 Let by contrary assume that Algorithm \ref{main_algo} doesn't add path $P_{j + 1}$ i.e. Algorithm \ref{main_algo} does not add any path. Clearly it may happen due to one of the two reasons i.e. either $t$ is not reachable from $s$ or number of paths discovered $= p$(line number-$4$ in Algorithm \ref{main_algo}) or $T_{j + 1} > CET(S_{j})$(line number-$6$ in Algorithm \ref{main_algo}).\\
 \textbf{Case 1(a):} ($t$ is not reachable from $s$)\\
 As CCRP is able to find path $P_{j + 1}$, $t$ is reachable from $s$. Contradiction!\\ 
 \textbf{Case 1(b):} (Number of paths discovered $=p$)\\
 It is clear from CCRP Algorithm given in section \ref{CCRP} that it does not discover more than $p$ paths as in each path at least one people will be evacuated. As CCRP finds path $P_{j+1}$, number of  paths discovered before discovery of $P_{j+1}$ by Algorithm \ref{main_algo} can't be more than $p-1$.\\
 \textbf{Case 2:} ( $T_{j + 1} > CET(S_{j})$)\\
 Just come back to the point when CCRP adds path $(P_{j + 1}, C_{j + 1})$ for the first time. It can happen only in phase $j + 1$. From  Lemma \ref{analysis_lemma}  $P_{i}$ is returned  in $T_{j + 1}-T_{i} , 1 \leq i \leq j < k$, iterations  before phase $j + 1$. As $P_{j + 1}$ discovered in phase $j + 1$ for the first time total number of people evacuated through $P_{i}$ before discovery of $P_{j + 1}$ is at least $T_{j + 1}-T_{i}$. As group size of path $P_{i}$ is $C_{i}$, total number of people evacuated before discovery of $P_{j + 1}$ is at least  $\sum_{i = 1}^{j}C_{i}(T_{j + 1} - T_{i})$. As CCRP adds the path $P_{j + 1}$ we can say that still there are people to be evacuated. Also from Observation \ref{group_obsv} virtual evacuees are added while or after addition of path $P_{k}$. So, total number of people evacuated before discovery of $P_{j + 1}$ is strictly less than $p$. Mathematically $\sum_{i = 1}^{j}C_{i}(T_{j + 1} - T_{i}) < p$, which implies $T_{j + 1} \leq CET(S_{j}) $. Contradiction!\\
 \textbf{Part 2:}
 Now one thing remains to prove is available capacity of the path $P_{j + 1}$ returned by Algorithm \ref{main_algo}  is also $C_{j + 1}$. If $P_{j + 1}$ doesn't share any node or edge with previously discovered path we are done. So, assume that there is some node or edge $x$  which is common to both $P_{j + 1}$ and some $P_{i}$, $1 \leq i \leq j$. Here we argue considering $x$ as a node and argument for $x$ as an edge is same. Let $t_{n}^{k}$ denotes time required to travel from $s$(source) to node $n$ via path $P_{k}$ with out intermediate delay. Observe that $t_{x}^{j+1} \geq t_{x}^{i}$. From observation $1$  $P_{j + 1}$ is discovered in phase $j + 1$ for the first time by CCRP algorithm. In phase $j + 1$ consider $A_{j+1}(T_{j + 1})$. $P_{i}$ has been discovered once before discovery of $P_{j + 1}$   with its destination arrival time $T_{j + 1}$ i.e. it has made a reservation of $C_{i}$ at $x$ for the time instance $t_{x}^{j+1}$ at node $x$. Now arrival time of evacuees via $P_{j + 1}$ to $x$ is also $t_{x}^{j + 1}$. At $t_{x}^{j+1}$ we can not use that capacity of $C_{i}$ for evacuees routing via $P_{j + 1}$. In other words as if  node $x$ has dedicated capacity of $C_{i}$ at time $t_{x}^{j + 1}$ for evacuees routing via $P_{i}$ and that can't be used by evacuees routing via  $P_{j + 1}$. Here we have not assumed anything on $i$ and $x$. For each such $i$ and $x$, $P_{j + 1}$ can't use the capacity of $C_{i}$ at time $t_{x}^{j+1}$ at node $x$. It is equivalent to permanently decrementing the capacity of such $x$'s by corresponding $C_{i}$, because from observation $1$ whenever $P_{j + 1}$ is discovered prior to that a reservation of $C_{i}$ must have been done at common node $x$(of $P_{i}$ and $P_{j + 1}$) by path $P_{i}$. Now come back to Algorithm \ref{main_algo}. By induction each path $P_{i}$, $i \leq j$ is returned with capacity $C_{i}$. We find path $P_{j + 1}$ by decrementing the capacity of each path by $C_{i}$ permanently. So, just before addition of $P_{j + 1}$ structure of the graph remains same w.r.t. capacity both in CCRP and Algorithm \ref{main_algo}. From this discussion we can say that capacity of path $P_{j + 1}$ returned by Algorithm \ref{main_algo} is $C_{j + 1}$.
\end{proof}

\begin{theorem}\label{analysis_main_thm}
\textit{The evacuation time of the solution given by Algorithm \ref{main_algo} is at most as that of the CCRP Algorithm for single source and single sink.}
\end{theorem}

\begin{proof}
Let $(P_{1}, C_{1}), (P_{2}, C_{2}), \ldots, (P_{k}, C_{k})$ be distinct paths (not necessarily edge-disjoint) from $s$ to $t$ in order of their transit time (neglecting delays) discovered by CCRP. By Lemma \ref{analysis_lemma}, Algorithm \ref{main_algo} also returns the same set of paths. From Observation \ref{main_obsv}, we can say that evacuation time of CCRP is $T_{Evac}^{CCRP} = T_{k} + \epsilon - 1$. Evacuation time of Algorithm \ref{main_algo} is $CET(S_{k})$. Also from Lemma \ref{analysis_lemma}, number of people that are evacuated through $P_{i}$ is $C_{i}(T_{k}-T_{i} +\epsilon)$. As All people have been evacuated we can write
 $\sum_{i = 1}^{k}{C_{i}(T_{k}-T_{i} + \epsilon)} \geq p$, which implies $T_{Evac}^{CCRP} \geq CET(S_{k})$.
\end{proof}

\begin{theorem}\label{analysis_upperbound_thm}
Upper bound on the evacuation time given by CCRP (hence by Algorithm \ref{main_algo}) for single source single sink network is $\left \lfloor \frac{p}{k} \right \rfloor + (n - 1)\tau  - 1$, where $p$ is the number of evacuees, $n$ is the number of nodes in the graph, $\tau$ is the maximum transit time of any edge and $k$ is the number of paths used by CCRP (and Algorithm \ref{main_algo}).
\end{theorem}

\begin{proof}
From Lemma \ref{analysis_lemma}, number of iterations executed by CCRP is $\sum_{i=1}^{k}(T_{k} - T_{i} + \epsilon) \le p$, as in each iteration at least one person will be evacuated. Hence, $T_{Evac}^{CCRP} \leq \left \lfloor \frac{p}{k}\right \rfloor + (n-1)\tau  - 1$.
\end{proof}

\section{Randomized Behavior Model of People}
\par The idea of combined evacuation time \cite{min2011evacuation} can be extended by considering probabilistic behavior of people. Suppose in an evacuation, people  do not follow the paths suggested by Algorithm \ref{main_algo} (or CCRP). Let's say with probability $\alpha>0$ a person follows suggested path and with probability $1-\alpha$ he follows the shortest path (to the nearest exit). In this situation, we have to redistribute people via various paths. 
If we suggest $x_{i}$ persons via $P_{i},i \neq 1$, then the number of persons who will follow $P_{i}$ and $P_{1}$ is $\alpha x_{i}$ and $(1-\alpha)x_{i}$ respectively (in expectation). The total number of people following $P_{1}$ and $P_{i}$ are $x_{1} + \sum_{i=2}^{k}(1-\alpha)x_{i}$ and $\alpha x_{i},i\neq 1$ respectively. 
Expected time at which the last person will arrive at destination via $P_{1}$ is $T_{1} + \frac{ x_{1} +\sum_{i = 2}^{k}(1-\alpha)x_{i}}{C_{1}}-1 $. Expected time at which last person will arrive at destination via $P_{i}$ is $T_{i} + \frac{\alpha x_{i}}{C_{i}} -1, i\neq 1$\\
Let the expected evacuation time in this scenario be $E[T]$. Now we can write,
\[E[T] = \max\left(T_{1} + \frac{(1-\alpha)n}{C_{1}}-1, \max_{2 \leq i \leq k}\left(T_{i} + \frac{\alpha x_{i}}{C_{i}} - 1\right)\right).\]
$E[T]$ will be minimum when it satisfies the following equation,
\begin{align}
E[T] &= T_{1} + \frac{x_{1}+\sum_{i = 2}^{k}(1-\alpha)x_{i}}{C_{1}}- 1 \nonumber \\
&= T_{i} + \frac{\alpha x_{i}}{C_{i}} - 1, 2 \le i \le k \label{eqn:min-et}.
\end{align}
where $\sum_{i = 1}^{k}x_{i} = n$ and $x_{i} \geq 0, \forall i$.
Solving the above equations we get,
 \begin{equation}
 E[T] = \frac{n + \sum_{i = 1}^{k}{C_{i}T_{i}}}{\sum_{i = 1}^{k}C_{i}} - 1 = CET(\{P_{1},P_{2},\ldots,P_{k}\}) \label{eqn:et}
 \end{equation}
  Expected evacuation time given by equation (\ref{eqn:et}) doesn't depend on $\alpha$. This is true and solution is feasible as long as $x_{1} \geq 0$. But it is not always the case, specifically when $(1 - \alpha)\sum_{i = 2}^{k}x_{i} > C_{1}(T  - T_{1} + 1)$. So, implicitly evacuation time is dependent on $\alpha$.
\par   In the following sections we give the algorithm that considers the randomized behavior of people along with analysis for expected evacuation time.

\subsection{Lower bound for expected evacuation time}
 On expectation $x_{1} + (1- \alpha) \sum_{i = 2}^{k}x_{i} = \alpha x_{1} + (1-\alpha)n$ number of  people will be evacuated via path $P_{1}$. This is minimum when $x_{1} = 0$ as $x_{1} \geq 0$. So, lower bound for expected evacuation time is $T_{1} + \frac{(1-\alpha)n}{C_{1}} - 1$.

\subsection{Algorithm for randomized behavior of people}
\textbf{Algorithm $2$:}
\begin{enumerate}
\item Run Algorithm \ref{main_algo}. Find CET and $x_{1},x_{2},\ldots,x_{k}$ using Equation (\ref{eqn:min-et}).
\item If $x_{1} \geq 0$ then quit; else go to step 3. In this case, the expected evacuation time = CET.
\item Assign $x'_{1}$ to $0$ and $x'_{i} = \frac{nx_{i}}{\sum_{j = 2}^{k}x_{j}}, \forall i \neq 1$. In this case, the expected evacuation time = $T_{1} + \frac{(1-\alpha)n}{C_{1}} - 1$.
\end{enumerate}

\begin{lemma}\label{random_lemma_1}
\textit{$x'_{i} < x_{i}$, $\forall i \neq 1$, and $\sum_{i=2}^{k}x'_{i} = n$.}
\end{lemma}

\begin{proof}
Directly follows from the algorithm.
\end{proof}

\begin{lemma}\label{random_lemma_2}
\textit{Above algorithm has a expected evacuation time of $CET(\{P_{1},P_{2},\ldots,P_{k}\})$ when it quits from step-$2$.}
\end{lemma}

\begin{proof}
In this case $x_{1} \geq 0$. From the equation-4 also we can observe that $x_{i} \geq 0, \forall i \neq 1$. Hence the solution is feasible. So, we can safely say that the expected evacuation time is $CET$. 
\end{proof}

\begin{lemma}\label{random_lemma_3}
\textit{Above algorithm has a expected evacuation time of $T_{1} + \frac{(1-\alpha)n}{C_{1}} - 1$ when it quits from step-$3$.}
\end{lemma}

\begin{proof}
In this case $x_{1} < 0$ and by Lemma \ref{random_lemma_1} $x'_{i} < x_{i}, i \neq 1$. For $i \neq 1$ $x'_{i}$ number of people are suggested path $P_{i}$. Hence $T_{i} + \frac{px'_{i}}{C_{i}} -1 < T_{i} + \frac{px_{i}}{C_{i}} -1 < CET$ , $i \neq 1$ and $T_{1} + \frac{(1-\alpha)n}{C_{1}} - 1 > CET$.
\end{proof}

\begin{theorem}\label{random_main_lemma}
\textit{In a single source single sink evacuation problem, if people follow the path suggested by Algorithm $2$ with probability $\alpha$, then the expected evacuation time is $\max(CET,T_{1} + \frac{(1-\alpha)n}{C_{1}}-1)$ and algorithm runs in $O(\min(n, p) \cdot n \log n)$ time.}
\end{theorem}

\begin{table*}[ht]
\tiny
\caption{Comparison of evacuation time and run time of SSEP and CCRP algorithms}
\begin{center}
\begin{tabular}{|c|c|c|c|c|c|c|c|}
\hline
\textbf{Number of Nodes} & \textbf{Number of Evacuees} & \multicolumn{2}{|c|}{\textbf{Evacuation Time}} & \multicolumn{2}{|c|}{\textbf{Run Time}} & \multicolumn{2}{|c|}{\textbf{Improvement in SSEP over CCRP $\left(\frac{CCRP}{SSEP}\right)$}} \\
\cline{3-8}
$(n)$ & $(p)$ & \textsc{SSEP} & \textsc{CCRP} & \textsc{SSEP} & \textsc{CCRP} & \textsc{Evacuation Time} & \textsc{Run Time} \\
\hline
100	&	3000	&	68	&	69	&	0.124	&	1.326	&	1.01	&	10.69	\\
500	&	5000	&	130	&	130	&	0.358	&	2.73	&	1.00	&	7.63	\\
1000	&	7000	&	155	&	156	&	1.014	&	14.586	&	1.01	&	14.38	\\
1500	&	9000	&	115	&	117	&	1.466	&	35.443	&	1.02	&	24.18	\\
2000	&	15000	&	661	&	661	&	1.622	&	29.016	&	1.00	&	17.89	\\
2500	&	25000	&	179	&	186	&	2.761	&	25.739	&	1.04	&	9.32	\\
5000	&	40000	&	903	&	903	&	3.899	&	93.521	&	1.00	&	23.99	\\
10000	&	65000	&	517	&	520	&	12.012	&	231.535	&	1.01	&	19.28	\\
15000	&	95000	&	1848	&	1853	&	14.025	&	336.946	&	1.00	&	24.02	\\
25000	&	100000	&	1126	&	1128	&	23.134	&	815.682	&	1.00	&	35.26	\\
50000	&	120000	&	1436	&	1446	&	46.69	&	1684.217	&	1.01	&	36.07	\\
100000	&	110000	&	1032	&	1044	&	93.4952	&	3016.3005	&	1.01	&	32.26	\\
500000	&	100000	&	1698	&	1720	&	344.341	&	11363.253	&	1.01	&	33.00	\\
\hline
\end{tabular}
\end{center}
\label{tab:experiments}
\end{table*}

\section{Experimental Results}
\subsection{Details of the Experiments}
We executed the SSEP and CCRP algorithms on a Dell Precision T7600 server having an Intel Xeon E5-2687W CPU running at 3.1 GHz with 8 cores (16 logical processors) and 128 GB RAM. The operating system is Microsoft Windows 7 Professional 64-bit edition. We used the C/C++ network analysis libraries \emph{igraph} and \emph{LEMON} to implement the algorithms. We used \emph{netgen} to generate synthetic graphs. The number of vertices in the graph varies from 100 to 500,000. The number of people varies from 3,000 to 120,000. The results are shown in Table \ref{tab:experiments}. The graphs are plotted on a log-log scale.

\begin{figure}
\begin{center}
\includegraphics[scale=0.5]{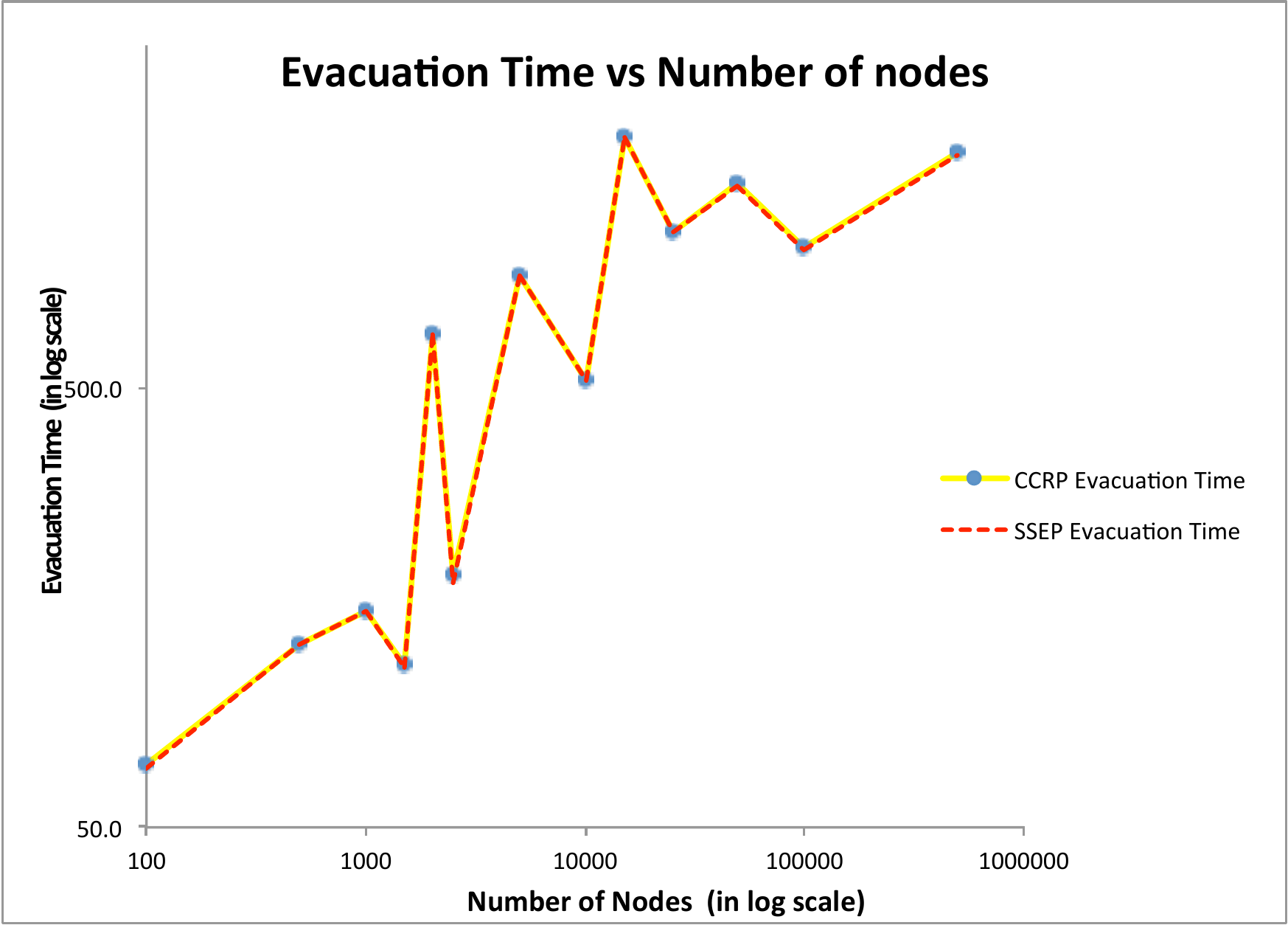}
\caption{Evacuation time vs number of nodes for SSEP and CCRP.}
\label{evacuation-time}
\end{center}
\end{figure}

\begin{figure}
\begin{center}
\includegraphics[scale=0.5]{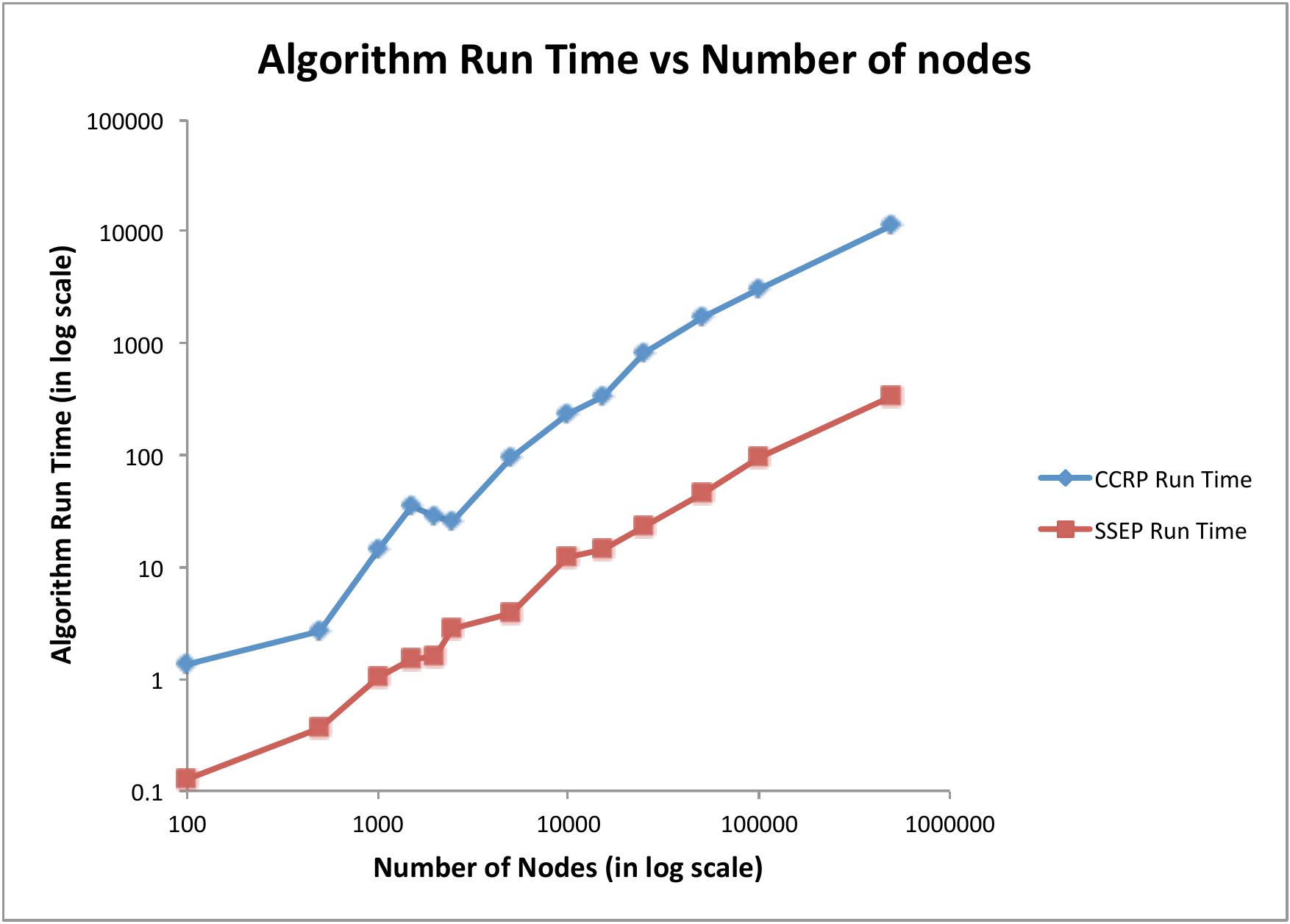}
\caption{Run time vs number of nodes for SSEP and CCRP.}
\label{run-time}
\end{center}
\end{figure}

\subsection{Results}
We show the variation of evacuation time and run time with number of nodes for SSEP and CCRP algorithms in Figure \ref{evacuation-time} and Figure \ref{run-time} respectively. From Figure \ref{evacuation-time}, we can see that the evacuation time of SSEP is at most that of CCRP. It is evident from Figure \ref{run-time} that the running time of SSEP is much lower than that of CCRP. Hence, for all these instances SSEP clearly outperforms CCRP with respect to both evacuation time and run time. The absolute and relative amount by which SSEP performs better than CCRP is shown in Table \ref{tab:experiments}. 

\section{Conclusion and Future Work}
In this paper, we have studied the evacuation route planning problem and given an improved algorithm for the single source single sink case. We theoretically showed that the SSEP algorithm performs better than the CCRP algorithm, both in terms of evacuation time and run time. This is also demonstrated by extensive experiments. We also analyzed a simple probabilistic behavior model of people. Here are some open problems which we would like to work in future.
\begin{itemize}
	\item Design a system for real time monitoring of evacuation in a building using our indoor localization app \cite{ahmed2015smartevactrak}.
	\item Extend this algorithm to the multiple source multiple sink case, and compare it's performance with CCRP and other algorithms.
	\item Develop a more sophisticated probabilistic behavior model of people for the case when they don't follow the routes suggested by the algorithm.
	\item Give good lower and upper bounds for the problem.
\end{itemize}

\bibliographystyle{plain}
\bibliography{evacuation}

\begin{thebibliography}{10}

\bibitem{ahmed2015smartevactrak}
Nasimuddim Ahmed, Avik Ghose, Amit~K Agrawal, Chirabrata Bhaumik, Vivek
  Chandel, and Abhinav Kumar.
\newblock {SmartEvacTrak: A people counting and coarse-level localization
  solution for efficient evacuation of large buildings}.
\newblock In {\em Pervasive Computing and Communication Workshops (PerCom
  Workshops), 2015 IEEE International Conference on}, pages 372--377. IEEE,
  2015.

\bibitem{dressler2010use}
Daniel Dressler, Martin Gro{\ss}, Jan-Philipp Kappmeier, Timon Kelter, Joscha
  Kulbatzki, Daniel Pl{\"u}mpe, Gordon Schlechter, Melanie Schmidt, Martin
  Skutella, and Sylvie Temme.
\newblock On the use of network flow techniques for assigning evacuees to
  exits.
\newblock {\em Procedia Engineering}, 3:205--215, 2010.

\bibitem{gupta2014efficient}
Ajay Gupta and Nandlal~L Sarda.
\newblock Efficient evacuation planning for large cities.
\newblock In {\em Database and Expert Systems Applications}, pages 211--225.
  Springer, 2014.

\bibitem{hamacher2001mathematical}
Horst~W Hamacher and Stevanus~A Tjandra.
\newblock {\em Mathematical modelling of evacuation problems: A state of art}.
\newblock Fraunhofer-Institut f{\"u}r Techno-und Wirtschaftsmathematik,
  Fraunhofer (ITWM), 2001.

\bibitem{hoppe1994polynomial}
Bruce Hoppe and {\'E}va Tardos.
\newblock Polynomial time algorithms for some evacuation problems.
\newblock In {\em Proceedings of the fifth annual ACM-SIAM symposium on
  Discrete algorithms}, pages 433--441. Society for Industrial and Applied
  Mathematics, 1994.

\bibitem{kim2008contraflow}
Sangho Kim, Shashi Shekhar, and Manki Min.
\newblock Contraflow transportation network reconfiguration for evacuation
  route planning.
\newblock {\em IEEE Transactions on Knowledge and Data Engineering},
  20(8):1115--1129, 2008.

\bibitem{lovs1998models}
Gunnar~G L{\o}vs.
\newblock Models of wayfinding in emergency evacuations.
\newblock {\em European journal of operational research}, 105(3):371--389,
  1998.

\bibitem{lu2005capacity}
Qingsong Lu, Betsy George, and Shashi Shekhar.
\newblock Capacity constrained routing algorithms for evacuation planning: A
  summary of results.
\newblock In {\em Advances in spatial and temporal databases}, pages 291--307.
  Springer, 2005.

\bibitem{min2012synchronized}
Manki Min.
\newblock Synchronized flow-based evacuation route planning.
\newblock In {\em Wireless Algorithms, Systems, and Applications}, pages
  411--422. Springer, 2012.

\bibitem{min2013maximum}
Manki Min and Jonguk Lee.
\newblock Maximum throughput flow-based contraflow evacuation routing
  algorithm.
\newblock In {\em Pervasive Computing and Communications Workshops (PERCOM
  Workshops), 2013 IEEE International Conference on}, pages 511--516. IEEE,
  2013.

\bibitem{min2011evacuation}
Manki Min and Bipin~C Neupane.
\newblock An evacuation planner algorithm in flat time graphs.
\newblock In {\em Proceedings of the 5th International Conference on Ubiquitous
  Information Management and Communication}, page~99. ACM, 2011.

\bibitem{skutella2009introduction}
Martin Skutella.
\newblock An introduction to network flows over time.
\newblock In {\em Research Trends in Combinatorial Optimization}, pages
  451--482. Springer, 2009.

\bibitem{AAAI159418}
Xuan Song, Quanshi Zhang, Yoshihide Sekimoto, Ryosuke Shibasaki, Nicholas~Jing
  Yuan, and Xing Xie.
\newblock A simulator of human emergency mobility following disasters:
  Knowledge transfer from big disaster data.
\newblock In {\em AAAI Conference on Artificial Intelligence}, 2015.

\bibitem{yin2009scalable}
Dafei Yin.
\newblock A scalable heuristic for evacuation planning in large road network.
\newblock In {\em Proceedings of the Second International Workshop on
  Computational Transportation Science}, pages 19--24. ACM, 2009.

\end{thebibliography}

\end{document}